\documentclass{amsart}

\usepackage{amsmath,amsthm}
\usepackage{tikz}
\usetikzlibrary{calc}

\newtheorem{conj}{Conjecture}
\newtheorem{thm}{Theorem}
\newtheorem{defn}{Definition}

\begin{document}

\thispagestyle{empty}

\markboth{Danny Rorabaugh}
	{A bound on a convexity measure for point sets}

\title{A bound on a convexity measure for point sets}

\author{Danny Rorabaugh}
\let\thefootnote\relax\footnote{University of South Carolina; rorabaug@email.sc.edu}

\begin{abstract}
A planar point set is in convex position precisely when it has a convex polygonization, that is, a polygonization with maximum interior angle measure at most $\pi$. We can thus talk about the convexity of a set of points in terms of the minimum, taken over all polygonizations, of the maximum interior angle. The main result presented here is a nontrivial combinatorial upper bound of this min-max value in terms of the number of points in the set. Motivated by a particular construction, we also pose a natural conjecture for the best upper bound.\\

\end{abstract}

\maketitle

\section{Introduction}

A simple planar polygon is {\it convex} provided the segment connecting any two points in the polygon is also contained in the polygon. Henceforth we assume that everything is in $\mathbb{R}^2$. Given a point set, a {\it polygonization} of the set is a simple polygon whose vertices are the points of the set. We say that a point set in general position -- that is, with no three points collinear -- is {\it in convex position} provided it has a convex polygonization. For a point set that is not in convex position, can we quantifiably describe the `convexity' of the point set?

In this paper we consider a new way to measure point set convexity, prove a bound on that measure in terms of the number of points, and present a natural conjecture for the best such bound. Clearly a point set in general position is convex if and only if there exists a polygonization with maximum interior angle-measure less than $\pi$. Suppose the maximum interior angle of a polygonization of a fixed non-convex set, minimized over all polygonizations, has measure $\pi + \varepsilon$. Then a smaller $\varepsilon$ value would represent being `closer' to convex position. The conjecture at the end of the paper implies that $\varepsilon$ is at most $2\pi - 2\pi/(n-1)$; we prove a weaker bound in Theorem \ref{Main}.

\subsection{Similar Work}

There exists research on several other perspectives of the convexity of a point set. Arkin, et al.\cite{AF02} introduce the concept of reflexivity, which counts the smallest number of reflex angles (i.e., interior angles with measure greater than $\pi$) in a polygonization of a point set. Despite the similar flavor to the min-max measure presented here, there appears to be no immediate relation between the two. For example, it is easy to construct a series of polygons with an increasing proportion of reflex angles, yet maximum interior angle measure approaching $\pi$. Likewise, polygons with one reflex angle can have the measure of that angle arbitrarily close to $2\pi$.

Alternatively the maximum area of a polygonization, in relation to the area of the convex hull, could be used to describe how convex a point set is. Zunich and Rosin\cite{ZR04} discuss various `area-based' measures of convexity; Fekete\cite{Fe00} analyzes the complexity of finding the polygonalization of maximum area.

The number of triangulations of a point set could be considered a measure of convexity, as seen in Welzl's\cite{We07} discussion of results for some particular non-convex arrangements of points. A Delaunay triangulation of a point set is known to maxiize the minimum angle of all angles in the triangulation\cite{Mu97}. If a connection exists between these triangulations and the exterior angles of polygonizations, then the known algorithms and results for Delauny triangulations could very well prove relevant to the present work.

We might define convexity of a point set by how the set decomposes into sets in convex position. Urabe and Hosono\cite{Ur96,HU01} studied the fewest number of parts in a partition of a point set into vertex sets of disjoint convex polygons; this was previously explored by Chazelle and Dobkin\cite{CD85} for fixed polygons rather than point sets. In a similar vein, Wu and Ding\cite{WD08} studied for point sets the maximum number of subsets in convex position that, though not necessarily disjoint, have no other points in their interior.

One quality of a convex polygon is the full visibility of the entire polygon from any interior point. Ever since the `art gallery problem' was posed by Klee in 1973, as related by O'Rourke\cite{OR87}, there has been extensive research into various forms of polygon visibility: by Rote\cite{Ro13} and Stern\cite{St89}, for example. A visibility measure or any other polygonal measure that is maximized or minimized by convex polygons could feasibly be used for point sets by optimizing over all possible polygonizations.

\section{Main Result}

\begin{defn}
Let $S$ be a set of $n > 3$ points in general position.
\begin{itemize}
\item An {\it extremal point} of $S$ is a point on the convex hull of $S$ and an {\it internal point} of $S$ is a non-extremal point.
\item The circle $C(S)$ that {\it circumscribes} S is the smallest circle containing $S$.
\end{itemize}
\end{defn}

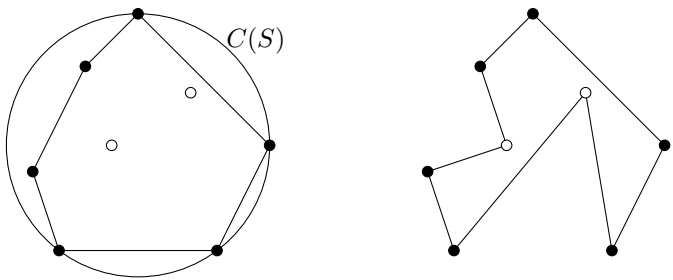
\begin{figure}[h] 
\center
\begin{tikzpicture}[scale=.35]
	\draw (0,5) -- (5,0) -- (3,-4) -- (-3,-4) -- (-4,-1) -- (-2,3) -- (0,5);
	\filldraw (5,0) circle (.2cm);
	\filldraw (3,-4) circle (.2cm);
	\draw[fill=white] (2,2) circle (.2cm);
	\filldraw (-2,3) circle (.2cm);
	\filldraw (-3,-4) circle (.2cm);
	\filldraw (0,5) circle (.2cm);
	\draw[fill=white] (-1,0) circle (.2cm);
	\filldraw (-4,-1) circle (.2cm);
	\draw (0,0) circle (5cm);
	\draw (3,4) node[right]{$C(S)$};

\begin{scope}[shift={(15,0)}]
	\draw (0,5) -- (5,0) -- (3,-4) -- (2,2) -- (-3,-4) -- (-4,-1) -- (-1,0) -- (-2,3) -- (0,5);
	\filldraw (5,0) circle (.2cm);
	\filldraw (3,-4) circle (.2cm);
	\draw[fill=white] (2,2) circle (.2cm);
	\filldraw (-2,3) circle (.2cm);
	\filldraw (-3,-4) circle (.2cm);
	\filldraw (0,5) circle (.2cm);
	\draw[fill=white] (-1,0) circle (.2cm);
	\filldraw (-4,-1) circle (.2cm);
\end{scope}
\end{tikzpicture}

\caption{A set $S$ of points with closed circles for extremal points and open circles for internal points. The left image shows the convex hull and $C(S)$; the right shows a polygonization of $S$.}
\label{polyg}
\end{figure}

\begin{thm} \label{Main}
Let $S$ be a set of $n > 3$ points in general position with $x < n$ extremal points. There exists a polygonization of $S$ such that every interior angle has measure at most $2\pi - \displaystyle \frac{\pi}{(n-1)(n-x)}$.
\end{thm}

\begin{proof}
The proof of the main theorem has two stages. First we will construct a set $\mathcal{P}(S)$ of polygonizations of $S$. Second we will show that one of the polygons in $\mathcal{P}(S)$ satisfies the desired property.

In the first stage, we will use each point $c$ on $C(S)\setminus S$ such that $S\cup\{c\}$ is still in general position. Since the line through any two points of $S$ intersects $C(S)$ twice, there are at most $2\binom{|S|}{2}$ points of $C(S)$ that are not as described. Removing these finitely many troublesome points, we label the remaining set $D(S)$. Observe for future use that $D(S)$ is the disjoint union of arcs of $C(S)$, the sum of whose (interior) angle measures is still $2\pi$.

Now given a point $d \in D(S)$, we will define an associated polygonization $P(d)$ on $S$ (illustrated in Fig. \ref{Pd}). Note that for every segment of the convex hull of $S$, all the interior points of the convex hull lie on the same side of the segment - or rather, the line containing the segment. To begin, we include in $P(d)$ every segment of the convex hull such that $d$ is on the same side as the interior points of the convex hull. By convexity of the convex hull, the segments included in $P(d)$ will be connected, forming a polygonal arc whose endpoints we label $x$ and $y$. Note that the angle $\angle xdy$ has in its interior all points of $S$ that are not on the polygonal arc. Next we assign an angular order to these points with respect to rotation around $d$ away from the ray $\overrightarrow{dx}$, and respectively label them $x_1, x_2, \ldots, x_k$. Setting $x=x_0$ and $y = x_{k+1}$ we include the line segment $\overline{x_ix_{i+1}}$ in $P(d)$ for each $i \in \{0,1, \cdots,k\}$, and so complete our polygonization $P(d)$.

\begin{figure}
\center
\begin{tikzpicture}[scale=.4]
	\draw[very thick] (0,5) -- (5,0) -- (3,-4) -- (-3,-4) -- (-4,-1);
	\draw[dashed] (-4,-1) -- (-2,3) -- (0,5);
	\filldraw (5,0) circle (.15cm);
	\filldraw (3,-4) circle (.15cm);
	\filldraw (2,2) circle (.15cm);
	\filldraw (-2,3) circle (.15cm);
	\filldraw (-3,-4) circle (.15cm);
	\filldraw (0,5) circle (.15cm) node[above]{$y$};
	\filldraw (-1,0) circle (.15cm);
	\filldraw (-4,-1) circle (.15cm) node[left]{$x$};
	\draw (0,0) circle (5cm);
	\filldraw[] (-4,3) circle (.2cm) node[left] {$d$};

\begin{scope}[shift={(15,0)}]
	\draw (0,5) -- (5,0) -- (3,-4) -- (-3,-4) -- (-4,-1);
	\draw[very thick] (-4,-1) -- (-1,0) -- (2,2) -- (-2,3) -- (0,5);
	\draw[->,dashed](-4,3) -- (-4,-5);
	\draw[->,dashed](-4,3) -- (2,6);
	\filldraw (5,0) circle (.15cm);
	\filldraw (3,-4) circle (.15cm);
	\filldraw (2,2) circle (.15cm);
	\filldraw (2,2) circle (.15cm) node[below]{$x_2$};
	\filldraw (-2,3) circle (.15cm) node[below]{$x_3$};
	\filldraw (-3,-4) circle (.15cm);
	\filldraw (0,5) circle (.15cm) node[above]{$x_4$};
	\filldraw (-1,0) circle (.15cm) node[above]{$x_1$};
	\filldraw (-4,-1) circle (.15cm) node[left]{$x_0$};
	\filldraw (-4,3) circle (.2cm) node[left] {$d$};
	\draw (0,-2) node {$P(d)$};
\end{scope}
\end{tikzpicture}

\caption{The construction of polygonization $P(d)$ of $S$. On the left is the polygonal arc associated with $d$; the right shows the ordering on the remaining vertices of $S$ and the completed polygon.}
\label{Pd}
\end{figure}
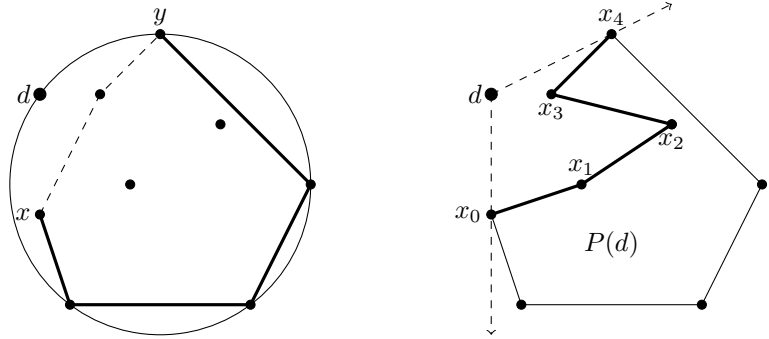

Our set of polygonizations of $S$ is $\mathcal{P}(S) = \{P(d) \mid d \in D(S)\}$. Even though $D(S)$ is uncountable, we can show that $\mathcal{P}(S)$ is finite. Consider distinct points $c,d \in D(S)$ such that $c$ and $d$ are on the same side of each of the $\binom{|S|}{2}$ segments connecting two points of $S$. By the construction above, $P(c) = P(d)$. We previously noted that $D(S)$ is the disjoint union of arcs of $C(S)$. No arc in $D(S)$ intersects with any line containing one of the $\binom{|S|}{2}$ aforementioned segments. Thus we have exactly one polygonization associated with each maximal arc of $D(S)$. Furthermore, the maximal arcs can be classified by an assigment of $\binom{|S|}{2}$ binary values. Since some assignments may be impossible and other assigments may be redundant, we do not have equality, but rather $|\mathcal{P}(S)| \leq 2^{\binom{|S|}{2}}$. The exact value is unimportant; we only need finitude for our next objective. We also need the following terminology.

\begin{defn}\mbox{}
Given an internal point $p$ of $S$, the {\it Potential Exterior Angles} -- \textsc{pea}s -- at $p$ are the $n-1$ angles formed by the $n-1$ rays from $p$ to the other points of $S$. (See Fig. \ref{PEAs}.)
\end{defn}

Given a point $d \in D(S)$, we find that $d$ is in the interior of every exterior angle of $P(d)$ with measure less than $\pi$. No such exterior angle of $P(d)$ can involve a segment from the polygonal arc of the former stage of construction, as the arc is part of the convex hull of $S$. Take an exterior angle formed by two consecutive segments chosen in the latter stage of construction of $P(d)$: if $d$ were not in the interior, the three points defining the angle would have been ordered differently.

Futhermore, note that the vertex of each exterior angle with measure less than $\pi$ is an internal point of $S$. Label as $s$ the vertex of an exterior angle of $P(d)$ with smallest measure. This angle is a composition of consecutive \textsc{pea}s at $s$. We will throw into the proverbial `pot' the \textsc{pea} $\varphi(d)$ at $s$ which contains $d$ in its interior. (See Fig. \ref{PEAs}.) Angle $\varphi(d)$ has two important properties. Property 1 is that the entire maximal arc of $D(S)$ containing $d$ is in the interior of $\varphi(d)$. Then from circle geometry, we know that the angle measure of the arc is at most twice the angle measure of $\varphi(d)$. Property 2 is that $\varphi(d)$ has angle measure at most as large as the smallest exterior angle-measure of $P(d)$.

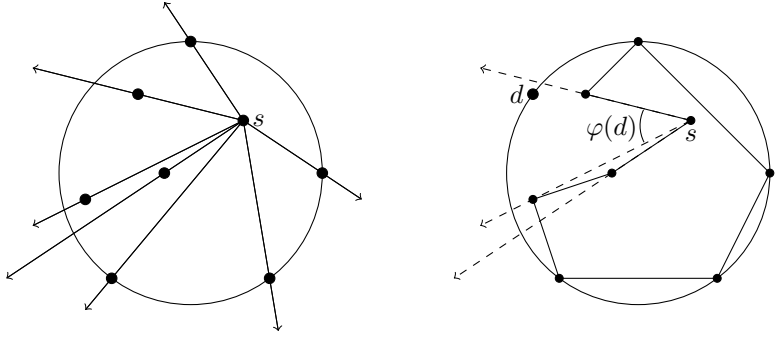
\begin{figure}
\center
\begin{tikzpicture}[scale=.35]
	\filldraw (2,2) circle (.2cm) node[right] {$s$};
	\filldraw (5,0) circle (.2cm);
	\filldraw[->] (2,2) -- (6.5,-1);
	\filldraw (3,-4) circle (.2cm);
	\filldraw[->] (2,2) -- (3.33,-6);
	\filldraw (-2,3) circle (.2cm);
	\filldraw[->] (2,2) -- (-6,4);
	\filldraw (-3,-4) circle (.2cm);
	\filldraw[->] (2,2) -- (-4,-5.2);
	\filldraw (0,5) circle (.2cm);
	\filldraw[->] (2,2) -- (-1,6.5);
	\filldraw (-1,0) circle (.2cm);
	\filldraw[->] (2,2) -- (-7,-4);
	\filldraw (-4,-1) circle (.2cm);
	\filldraw[->] (2,2) -- (-6,-2);
	\draw (0,0) circle (5cm);

\begin{scope}[shift={(17,0)}]
	\draw (0,5) -- (5,0) -- (3,-4) -- (-3,-4) -- (-4,-1) -- (-1,0) -- (2,2) -- (-2,3) -- (0,5);
	\filldraw (5,0) circle (.15cm);
	\filldraw (3,-4) circle (.15cm);
	\filldraw (2,2) circle (.15cm) node[below]{$s$};
	\filldraw (-2,3) circle (.15cm);
	\filldraw (-3,-4) circle (.15cm);
	\filldraw (0,5) circle (.15cm);
	\filldraw (-1,0) circle (.15cm);
	\filldraw (-4,-1) circle (.15cm);
	\draw (0,0) circle (5cm);
	\filldraw[] (-4,3) circle (.2cm) node[left] {$d$};
	\draw[->,dashed](2,2) -- (-6,4);
	\draw[->,dashed](2,2) -- (-6,-2);
	\draw[->,dashed](2,2) -- (-7,-4);
	\draw (.2,2.4) arc (155:207:1.4);
	\draw (-1,1.7) node {$\varphi(d)$};
\end{scope}
\end{tikzpicture}

\caption{Picking \textsc{pea}s. Left shows the rays that define the \textsc{pea}s at point $s$; right shows the two \textsc{pea}s that compose the minimum-measure exterior angle of $P(d)$, $\varphi(d)$ being the \textsc{pea} containing $d$.}

\label{PEAs}
\end{figure}

For each $d \in D(S)$, we drop the \textsc{pea} $\varphi(d)$ in the pot. Since there are $n-x$ internal points and $n-1$ \textsc{pea}s per internal point, there is a clear upper bound of $(n-1)(n-x)$ for the number of \textsc{pea}s in the pot (ignoring multiplicity). For each $\varphi$ in the pot, there might be multiple maximal arcs of $D(S)$ such that $\varphi = \varphi(d)$ for each $d$ in the arc. But the arcs are disjoint, so the sum of their angle measures is at most twice the measure of $\varphi$, by Property 1. Let $m$ be the maximum angle measure of the \textsc{pea}s in the pot. Recalling that the sum of the measures of the maximal arcs in $D(S)$ is $2\pi$, we get the  inequality: $2\pi \leq (2m)(n-1)(n-x)$. Thus gives the following upper bound on $m$:
$$m \geq \frac{\pi}{(n-1)(n-x)}.$$

Now we ladel out from the pot a \textsc{pea} $\theta$ whose angle measure is $m$, and take a polygon $P \in \mathcal{P}(S)$ whose minimum-measure angle consists of $\theta$, and possibly other \textsc{pea}s. By Property 2, the largest interior angle of $P$ is at most $2\pi - m$, and
$$2\pi - m \leq 2\pi - \frac{\pi}{(n-1)(n-x)}.$$

\end{proof}

\subsection{A Note on Complexity}
The computational complexity of obtaining a polygonization satisfying the bound of Theorem \ref{Main} can be readily extracted from the constructive first stage of the proof. There we identified a set of at most $\binom{n}{2}$ polygonizations, one of which has the desired property. The most computationally expensive step of constructing each polygonization is sorting the points not found on the initial polygonal arc, which can be done in $O(n\log n)$ time. Therefore, constructing the entire set of candidates can be done in $O(n^3 \log n)$ time. This absorbs the $O(n^3)$ time required to measuring the interior angles of all these polygonizations and identify the winner.

\section{A Natural Conjecture}

The proof of Theorem \ref{Main} appears loose in that it allows for adding all the \textsc{pea}s to the soup, which is surely overkill. It seams that the strict lower bound on $m$ should be on the order of $1/n$, as posited in the following natural conjecture. The author can, with messy details, decrease the bound on the number of \textsc{pea}s in the pot by a quantity of order $n$. However, this effort is asymptotically insignificant.

\begin{conj}
There exists a polygonization of $S$ with every interior angle of measure at most $2\pi-\frac{2\pi}{n-1}$. If true, this bound is tight, as demonstrated by the following construction: $n-1$ points equally spaced around a circle with the $n^{th}$ point at the center.
\end{conj}

\section*{Acknowledgements}

Many thanks to my advisor Professor Joshua Cooper and my colleagues Travis Johnston and Heather Smith for their feedback at various stages.
Also to Professor \'{E}va Czabarka for the discrete geometry course that led to this train of thought.

\appendix

%

\end{document}